\documentclass[12pt]{article}
\usepackage{amsthm,amsmath,amssymb,graphicx,mathrsfs,cite,array,graphicx,cases,color,setspace}
\usepackage{mathrsfs,cite,array}
\usepackage{cases,subfig,url,fullpage}
\usepackage{graphicx,mathrsfs}

\newcommand\var{{\rm var}}

\newtheorem{thm}{Theorem}[section]
\newtheorem{cor}[thm]{Corollary}
\newtheorem{conj}[thm]{Conjecture}

\newtheorem{lemma}[thm]{Lemma}

\title{Discriminating Power of Centrality Measures\\ (working paper)}

\author{
M. Puck Rombach\footnotemark[1]
\and
Mason A. Porter\footnotemark[3]
}

\begin{document}

\maketitle

\renewcommand{\thefootnote}{\fnsymbol{footnote}}
\footnotetext[1]{Oxford Centre for Industrial and Applied Mathematics, Mathematical Institute, University of Oxford, {\tt rombach@maths.ox.ac.uk}}
\footnotetext[2]{Oxford Centre for Industrial and Applied Mathematics, Mathematical Institute and
CABDyN Complexity Centre, University of Oxford, {\tt porterm@maths.ox.ac.uk}}

\begin{abstract}

The calculation of \emph{centrality measures} is common practice in the study of networks, as they attempt to quantify the importance of individual vertices, edges, or other components.  Different centralities attempt to measure importance in different ways. In this paper, we examine a conjecture posed by E. Estrada regarding the ability of several measures to distinguish the vertices of networks. Estrada conjectured \cite{Est05,Est13} that if all vertices of a graph have the same subgraph centrality, then all vertices must also have the same degree, eigenvector, closeness, and betweenness centralities. We provide a counterexample for the latter two centrality measures and propose a revised conjecture.

\end{abstract}

\section{Introduction}

Many problems in the study of networks focus on flows along the edges of a network \cite{newman2010}. Although flows such as the transportation of physical objects or the spread of disease are conservative, the flow of information tends to diminish as it moves through a network, and the influence of one vertex on another decreases with increasing network distance between them \cite{lerman}. 

Such considerations have sparked an interest in network \emph{communicability}, in which a perturbation of one vertex is felt by the other vertices with differing intensities~\cite{Est08}. These intensities depend on all paths between a pair of vertices, though longer paths have smaller influence. This idea is expressed in a general communicability function
\begin{equation} \label{comm}
	c_i=\sum_{k=0}^\infty \sum_{j=1}^n c_k (A^k)_{ij}\,,
\end{equation}
where $A$ is the adjacency matrix---whose entries are 1 if vertices $i$ and $j$ are connected to each other and 0 if they are not---and $n$ is the total number of vertices in the network.  The networks that we consider in the present paper are unweighted, undirected graphs with no self-edges and no multi-edges. A special version of communicability is the Katz centrality \cite{Katz53}, for which $c_k=\alpha^k$ (also see Bonacich power centrality~\cite{Bon87}). Reference \cite{Est12comm} gives an extensive review of network communicability.

Note that one can compute the \emph{self-communicability} of vertices:
\begin{equation}\label{karma}
	c_{ii}=\sum_{k=0}^\infty c_k (A^k)_{ii}\,.
\end{equation}
From a sociological perspective, one can construe $c_{ii}$ as the ``karma" of vertex $i$, as it measures its ability to receive its own influence via the graph. 

\emph{Subgraph centrality}, which was introduced in Ref.~\cite{Est05}, is an example of a communicability measure and is defined as
\begin{equation}\label{sum}
	sc_i=\sum_{k=0}^{\infty} \frac{\left( A^k \right)_{ii}}{k!}\,.
\end{equation}
It is a special case of self-communicability in which the weighting function is $c_k=1/k!$ (which has computational advantages). The sum in (\ref{sum}) is bounded by $\sum_{k=0}^{\infty} \frac{ \lambda_1^k }{k!}=e^\lambda$, so it converges. Subgraph centrality can be computed from the spectrum of the adjacency matrix \cite{Est05}:
\begin{equation*}
	sc_i=\sum_{j=1}^n (v_j^i)^2 e^{\lambda_j}\,,
\end{equation*}	
where $\lambda_j$ is the $j$th eigenvalue of $A$ and $v_j^i$ is the $i$th component of its associated eigenvector. We will denote the vector containing all vertex subgraph centralities as $sc_G$ and use similar notation for other centrality measures.

\section{Centrality Measures and Their Power To Discriminate}

There exist many methods to measure the \emph{centrality}, or importance, of vertices or other components in a graph \cite{newman2010}.  We consider measures of vertex centrality. Different methods for measuring vertex centrality are appropriate for different types of networks, dynamical processes on networks, and applications. 

In this paper, we consider subgraph centrality and four of the traditional notions of centrality: degree, closeness, betweenness, and eigenvector centralities. The \emph{degree centrality} (i.e., degree) of vertex $i$ is
\begin{equation}
	d_i=\sum_{j=1}^n A_{ij} 
\end{equation}
and is simply the number of edges incident on the vertex.

The \emph{eigenvector centrality} of a vertex $i$ is the value of the $i${th} entry of the leading eigenvector of the adjacency matrix \cite{bonacich72}:
\begin{equation}
	ec_i = \frac{1}{\lambda_0} \sum_{j \in \Gamma(i)} ec_j\,.
\end{equation}	
One way to think of the eigenvector centrality of a vertex $i$ is as the sum over $k$ of the number of distinct paths of length $k$ that start at vertex $i$ normalized by the total number of paths of length $k$ in a graph.

\emph{Closeness centrality} has been defined in multiple ways in the literature, but it is always a function of the mean length of the shortest paths from a vertex to each of the other vertices in a graph \cite{Sabidussi1966,Silva08}. We use the definition~\cite{Est11,Free79}
\begin{equation}\label{close}
	cc_i=\frac{n-1}{\sum_{j \in V(G)} P(i,j)}\,,
\end{equation}
where $P(i,j)$ is the length of a shortest (i.e. geodesic) path between $i$ and $j$.  The definition (\ref{close}) is the same as the one used in Ref.~\cite{Est13}.

The geodesic \emph{betweenness centrality} of a vertex measures how often it lies on shortest paths between pairs of nodes. It is defined as \cite{newman2010}
\begin{equation}
	bc_i = \sum _{j,k \in V(G) \backslash i} \frac{P^*_{jk}(i)}{P^*_{jk}}\,,
\end{equation}
where $P^*_{jk}$ is the number of distinct shortest paths running between vertices $j$ and $k$, and $P^*_{jk}(i)$ is the number of such paths that include vertex $i$.

It was noted in Ref.~\cite{Est13} that centrality measures such as degree, closeness, betweenness, and eigenvector centralities are unable on certain graphs to distinguish between any of its vertices---even for graphs that are not \emph{vertex-transitive}. By definition, a graph is ``vertex-transitive" if its vertices are indistinguishable: in other words, for all pairs of vertices $i$ and $j$, there exists an automorphism $f: V(G) \to V(G)$ that maps $i$ to $j$.

In Fig.~\ref{frucht}, we show the Frucht graph, which is a cubic (i.e., 3-regular) graph on 12 vertices that has no nontrivial symmetries (i.e., it has no automorphisms other than the identity). This is the smallest regular graph \cite{Bar69} with this property. Because degree centrality and eigenvector centrality can be described purely in terms of the degrees of vertices and those of their neighbors, these two measures are unable to distinguish between the vertices of the Frucht graph. 

In Fig.~\ref{betfail}, we show a graph on 6 vertices that is clearly not vertex-transitive (as there exist vertices of degrees 3 and 4) and for which all vertices have the same value for betweenness centrality~\cite{Est13}.  Such graphs are called \emph{betweenness-self-centric} (or are said to have a \emph{betweenness centralization} equal to $0$), and Ref.~\cite{Gago12} includes several examples of such graphs.

In Fig.~\ref{closfail}, we show a graph on 8 vertices that is not vertex-transitive (three vertices form part of a triangle, but the other five do not) and for which closeness centrality is unable to distinguish the vertices from each another~\cite{Est13}. In this example, degree centrality and eigenvector centrality are also unable to distinguish vertices from each other.

\begin{center}
\begin{figure}[!htbp]
\begin{minipage}[c]{\linewidth}
\begin{center}
\includegraphics[scale=.65]{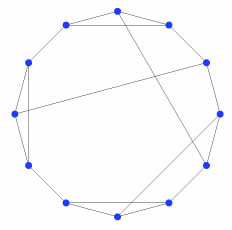}
\end{center}
\end{minipage}
\caption{The Frucht graph, which is an example of a asymmetric graph for which degree and eigenvector centrality are unable to distinguish the vertices from each other.
} 

\label{frucht}
\end{figure}
\end{center}

\begin{center}
\begin{figure}[!htbp]
\begin{minipage}[c]{\linewidth}
\begin{center}
\includegraphics[scale=.65]{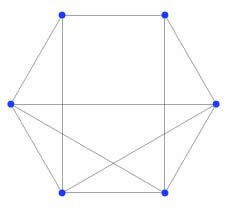}
\end{center}
\end{minipage}
\caption{A graph on 6 vertices that is not vertex-transitive and on which betweenness centrality is unable to distinguish the vertices from each other.}\label{betfail}
\end{figure}
\end{center}

\begin{center}
\begin{figure}[!htbp]
\begin{minipage}[c]{\linewidth}
\begin{center}
\includegraphics[scale=.65]{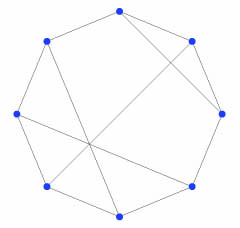}
\end{center}
\end{minipage}
\caption{A graph on 8 vertices that is not vertex-transitive and on which closeness, degree, and eigenvector centralities are unable to distinguish the vertices from each other.}\label{closfail}
\end{figure}
\end{center}

\section{Estrada's Conjectures}\label{est-section}

E. Estrada posed the following two conjectures concerning the discriminating power of subgraph centrality.
\begin{conj}\label{estconj1}~\cite{Est05}\quad
	$\var(sc_G)=0$ $\Rightarrow$ $\var(d_G)=\var(ec_G)=\var(cc_G)=\var(bc_G)=0$.
\end{conj}

\begin{conj}\label{estconj2}~\cite{Est13}\quad
	$\var(sc_G)=0$ $\Leftrightarrow$ $G$ is walk-regular.
\end{conj}

The quantity $\var(c_G)$ denotes the variance of $c_G$, where we have interpreted the vector of centrality values as a set. When $\var(c_G)=0$, all vertices have the same centrality value. 

In this section, we show that Conjecture \ref{estconj1}, which was originally posed in Ref.~\cite{Est05}, is at least partially false by giving a counterexample for betweenness centrality and closeness centrality. 

A graph is called \emph{walk-regular} if, for every $k \in \mathbb{N}$, the number of closed walks in $G$ of length $k$ that start at vertex $i$ is the same for all $i \in \{1, \ldots, n\}$. Because $sc_i$ depends only on the number of \emph{closed walks} (i.e., walks that start and end at the same vertex) of length $k$ starting at $i$, we obtain the following corollary.
\begin{cor}\label{walkobvious}\cite{Est05}\quad
	$G$ is walk-regular $\Rightarrow$ $\var(sc_G)=0$\,.
\end{cor}
A graph is walk-regular if and only if $\mathbb{P}_i(Z(k)=i)=\mathbb{P}_j(Z(k)=j)$ for all $k$ and for all vertices $i,j \in V(G)$~\cite{Geo12}, where $\mathbb{P}_i(Z(k)=i)$ denotes the probability that an ordinary random walk starting at $i$ is back at its initial point after $k$ steps. All vertex-transitive graphs are walk-regular, but the converse is not true~\cite{God80}.

A \emph{distance-regular} graph is a $k$-regular graph with diameter $d$ and a set $\{b_0,b_1,\ldots,b_{d-1},c_1,c_2,\ldots,c_d\}$, with $b_0=k$ and $c_1=1$, such that all vertices $i,j \in V(G)$ with distance (i.e.,  shortest path length) $r$ between them satisfy the following property: the number of vertices in $G_{r-1}(i)$ that are adjacent to $j$ is $c_r$, and the number of vertices in $G_{r+1}(i)$ that are adjacent to $j$ is $b_r$~\cite{Brou89,Pau07}. All distance-regular graphs are walk-regular, but a distance-regular graph need not be vertex-transitive~\cite{Cam03}. 

There exist walk-regular graphs that are neither vertex-transitive nor distance-regular~\cite{God80}. We show an example of such a graph in Fig.~\ref{wr-vt-dr}.\footnote{Note that the paper \cite{God80} has a mistake in the plot of this graph. We plot the correct graph in Fig.~\ref{wr-vt-dr}, and one can also find the correct graph in the erratum appended to the paper at \url{http://cs.anu.edu.au/people/bdm/papers/WalkRegular.jpg}.} This graph on 12 vertices, which we call the \emph{Godsil-McKay graph}, gives a counterexample to Conjecture \ref{estconj1}. Its centrality vectors are
\begin{align*}
	sc_{GM} &\approx 6.7035 \times [1 ,\: 1 ,\: 1 ,\: 1 ,\: 1 ,\: 1 ,\: 1 ,\: 1 ,\: 1 ,\: 1 ,\: 1 ,\: 1]\,,\\
	bc_{GM} &= [7 ,\: 8 ,\: 7 ,\: 8 ,\: 7 ,\: 7 ,\: 7 ,\: 7 ,\: 8 ,\: 7 ,\: 8 ,\: 7]\,,\\
	cc_{GM} &= \left[\frac{1}{18} ,\: \frac{1}{19} ,\: \frac{1}{18} ,\: \frac{1}{19} ,\: \frac{1}{18} ,\: \frac{1}{18} ,\: \frac{1}{18} ,\: \frac{1}{18} ,\: \frac{1}{19} ,\: \frac{1}{18} ,\: \frac{1}{19} ,\: \frac{1}{18} \:\right]\,.
\end{align*}
In Fig.~\ref{wr-vt-dr}, we use gray diamonds to designate the vertices with a betweenness centrality of 8 and a closeness centrality of ${1}/{19}$; we use blue disks to designate vertices with a betweenness centrality of 7 and a closeness centrality of ${1}/{18}$. Figure \ref{GGM} shows three more examples of graphs that are walk-regular, but neither distance-regular nor vertex-transitive. Betweenness centrality can distinguish all three of these examples, whereas closeness centrality can distinguish only two of them.

\begin{center}
\begin{figure}[!ht]
\begin{minipage}[c]{\linewidth}
\begin{center}
\includegraphics[scale=.65]{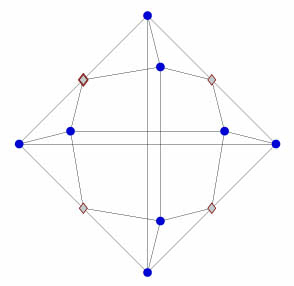}
\end{center}
\end{minipage}
\caption{The Godsil-McKay graph, which is walk-regular but neither vertex-transitive nor distance-regular \cite{God80}.}\label{wr-vt-dr}
\end{figure}
\end{center}

\begin{figure}[!ht]
\begin{center}
\begin{minipage}[c]{0.3\linewidth}
\centering
\includegraphics[scale=.4]{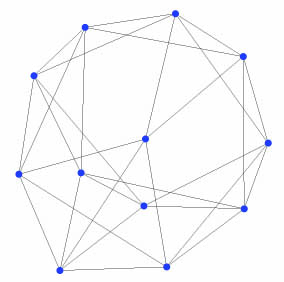}\\(a)
\end{minipage} \hspace*{3pt}
\begin{minipage}[c]{0.3\linewidth}
\centering
\includegraphics[scale=.4]{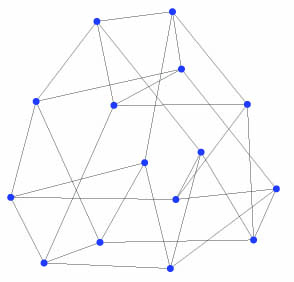}\\(b)
\end{minipage} \hspace*{3pt}
\begin{minipage}[c]{0.3\linewidth}
\centering
\includegraphics[scale=.4]{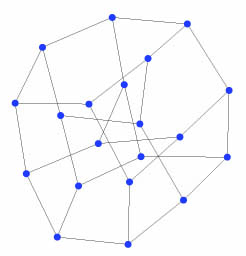}\\(c)
\end{minipage} \hspace*{3pt}
\end{center}
\caption{Three more examples of graphs that are walk-regular, but neither distance-regular nor vertex-transitive: (a) a 5-regular graph on 12 vertices, (b) a 4-regular graph on 15 vertices and (c) a 3-regular graph on 20 vertices.  Betweenness centrality can distinguish (a), (b), and (c). Closeness centrality can distinguish (b) and (c).
}\label{GGM}
\end{figure}

\begin{figure}[!ht]
\begin{center}
\begin{minipage}[c]{0.45\linewidth}
\centering
\includegraphics[scale=.25]{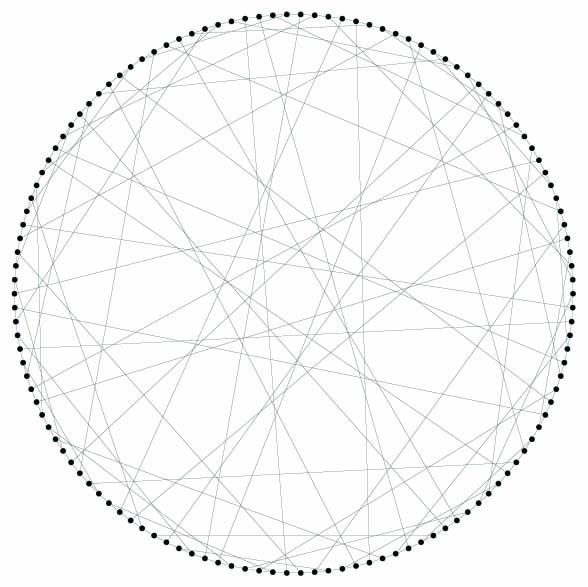}\\(a)
\end{minipage} \hspace*{3pt}
\begin{minipage}[c]{0.45\linewidth}
\centering
\includegraphics[scale=.25]{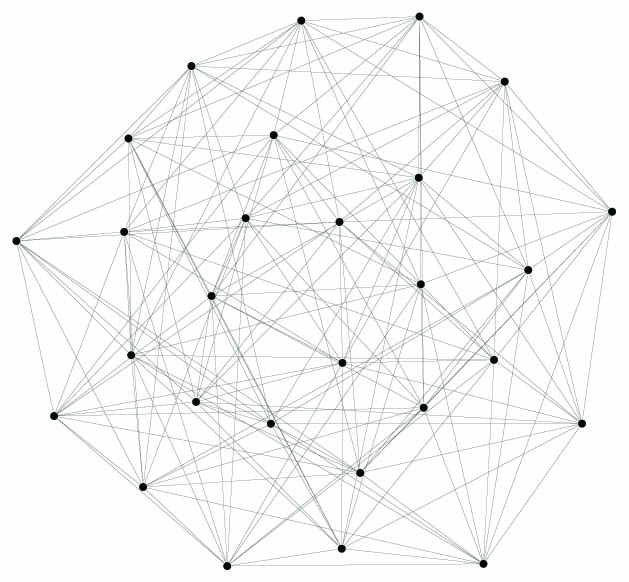}\\(b)
\end{minipage} \hspace*{3pt}
\end{center}
\caption{Two examples of graphs that are distance-regular (and hence also walk-regular) but not vertex-transitive: (a) the Tutte 12-cage \cite{Pegg} and (b) the first Chang graph \cite{Weiss}.}\label{tuttechang}
\end{figure}

In Fig.~\ref{tuttechang}, we give two examples of graphs that are distance-regular (and hence walk-regular) but not vertex-transitive. The examples are the Tutte 12-cage~\cite{Tits59,Ben66,Pegg} and one of the Chang graphs~\cite{Chang59,Weiss}. The vertices of these graphs cannot be distinguished by degree, eigenvector, closeness, betweenness, or subgraph centralities. The Tutte-12 cage and the Chang graphs are the only known distance-regular graphs that are not also vertex-transitive.\footnote{See, e.g., the discussion at \url{http://mathoverflow.net/questions/106589/is-every-distance-regular-graph-vertex-transitive}}. Additionally, it has been proven that betweenness centrality cannot distinguish distance-regular graphs. 
\begin{lemma}\label{distbet}~\cite{Gago12}\quad
	$G$ is distance-regular $\Rightarrow$ $\var(bc_G)=0$\,.
\end{lemma}
In fact, this is also true for closeness centrality.
	\begin{lemma}\label{distclos}$G$ is distance-regular $\Rightarrow$ $\var(cc_G)=0$\,.
\end{lemma}
\begin{proof} Let $\Gamma^r(i)=\{ j \vert \delta(i,j)=r \}$, where $r \in \{1, \ldots, d\}$ and $\delta(i,j)$ is the distance between vertex $i$ and $j$. This is a generalization of the neighborhood $\Gamma(i)=\{ j \vert ij \in E(G) \}$. Because we are only considering connected graphs, every $j \in V(G) \setminus i$ is associated with a finite positive integer $\delta(i,j)$. Therefore, the sets $\Gamma^r(i)$ and $\Gamma^r(j)$ are pairwise disjoint for any $i$ and $j$ (with $i \neq j$), and $V(G) = i \cup \Gamma^1(i) \cup \ldots \cup \Gamma^d(i)$. Given node $i$, we consider the subgraph $B_r$ of $G$ that consists of the union of $\Gamma^r(i)$ and $\Gamma^{r+1}(i)$ and all edges between the two sets.  We also recall that a \emph{biregular graph} is a bipartite graph in which all vertices in the same partite set have the same degree. Let $j$ be an arbitrary vertex in $\Gamma^{r}(i)$ and let $k$ be an arbitrary vertex in $\Gamma^{r+1}(i)$. From the definition of a distance-regular graph, we know that
\begin{align}
		\vert \Gamma (j) \cap \Gamma^{r+1} (i) \vert &= b_r\,, \notag  \\ 
		\vert \Gamma (k) \cap \Gamma^r (i) \vert &= c_{r+1}\,.
\end{align} 
Therefore, $B_r$ is a biregular graph.  Additionally, $\vert E(B_r) \vert = b_r \vert \Gamma^{r} (i) \vert = c_{r+1} \vert \Gamma^{r+1} (i) \vert$ for all $B_r$, so
\begin{equation}
	\vert \Gamma^r (i) \vert = \frac{b_{r-1}b_{r-2}\ldots b_0}{c_r c_{r-1} \ldots c_1}
\end{equation}
for each $i$. Because 
\begin{equation}
	cc_i=\frac{n-1}{\sum_{j \in V(G)} P(i,j)}=\frac{n-1}{\sum_{r=1}^d r \vert \Gamma^r (i)\vert }
\end{equation}
depends only on $\vert \Gamma^r(i) \vert$ (where $r \in \{1, \ldots, d\}$), which is independent of the node $i$, it follows that $\var(cc_G)=0$ for distance-regular graphs. 
\end{proof}

The converse of Lemma~\ref{distclos} is not true (see the counterexample in Fig.~\ref{closfail}).

\section{Conclusion}

Based on the results in Section \ref{est-section}, we propose the following conjecture, which we call the ``Modified Estrada Conjecture". 
\begin{conj}\label{newdiscconj}
$\var(sc_G)=0$ $\Leftrightarrow$ $G$ is walk-regular, and \\
	$\var(sc_G)=\var(bc_G)=\var(cc_G)=0$ $\Leftrightarrow$ $G$ is distance-regular.
\end{conj}
Note that the $\Leftarrow$ parts of the statements in Conjecture \ref{newdiscconj} are known to be true because of Corollary \ref{walkobvious} and Lemmas \ref{distbet} and \ref{distclos} (and the fact that all distance-regular graphs are walk-regular).

\section*{Acknowledgements}

This work was funded by the James S. McDonnell Foundation (\#220020177).  We thank James Fowler for helpful discussions and Brendan McKay for sending us the examples of walk-regular, non-distance-regular, non-vertex transitive graphs in Fig.~\ref{GGM}.  We thank Ernesto Estrada for sending us a copy of Ref.~\cite{Est13} prior to its publication.

\bibliography{refs5}        
\bibliographystyle{plain}

\end{document}